\documentclass{article}
\usepackage{kr}

\usepackage{times}
\usepackage{soul}
\usepackage{url}
\usepackage[hidelinks]{hyperref}
\usepackage[utf8]{inputenc}
\usepackage[small]{caption}
\usepackage{graphicx}
\usepackage{amsmath}
\usepackage{amsthm}
\usepackage{booktabs}
\usepackage{algorithm}
\usepackage{algorithmic}
\newtheorem{definition}{Definition}
\newtheorem{theorem}{Theorem}

\usepackage{color}

\title{The Complexity of Data-Driven Norm Synthesis and Revision}

\author{%
Davide Dell'Anna\and
Natasha Alechina\and
Brian Logan\and
Maarten L\"{o}ffler\and\\
Fabiano Dalpiaz\and
Mehdi Dastani\\
\affiliations
Utrecht University
\emails
\{d.dellanna, n.a.alechina, b.s.logan, m.loffler, f.dalpiaz, m.m.dastani\}@uu.nl
}
       
\begin{document}

\maketitle

\begin{abstract}
Norms have been widely proposed as a way of coordinating and controlling the activities of agents in a multi-agent system (MAS). A norm specifies the behaviour an agent should follow in order to achieve the objective of the MAS. However, designing norms to achieve a particular system objective can be difficult, particularly when there is no direct link between the language in which the system objective is stated and the language in which the norms can be expressed. In this paper, we consider the problem of synthesising a norm from traces of agent behaviour, where each trace is labelled with whether the behaviour satisfies the system objective. We show that the norm synthesis problem is NP-complete.
\end{abstract}

\section{Introduction}

There has been a considerable amount of work on using norms to coordinate the activities of agents in a multi-agent system (MAS) \cite{Chopra//:18a}. Norms can be viewed as standards of behaviour which specify that certain states or sequences of actions in a MAS should occur (obligations) or should not occur (prohibitions) in order for the objective of the MAS 
to be realized \cite{DBLP:conf/kr/BoellaT04}. We focus on conditional norms with deadlines which express behavioral properties~\cite{TDM2009}.
Conditional norms are triggered (detached) in certain states of the MAS and have a temporal dimension specified by a deadline. The satisfaction or violation of a detached norm depends on whether the behaviour of the agent brings about a specified state before a state in which the deadline condition is true. Conditional norms are implemented in a MAS through 
enforcement. That is, violation of a norm results in the violating agent incurring a sanction, e.g., a fine; see, e.g., \cite{Alechina//:13c,dell2020runtime} for how to determine an appropriate level of sanction.

For many applications it is assumed that the MAS developer will design an appropriate norm to realise the system objective. However, this can be difficult, particularly when the internals of the agents are unknown, e.g., in the case of open MAS~\cite{artikis2001formal}, and when there is no direct connection between the language in which the system objective is stated and the language in which norms can be expressed. For example, one 
objective of a traffic system may be to avoid traffic collisions, but `not colliding' is not a property under direct agent control, and prohibition of collisions cannot be stated as a norm. 
A poorly designed norm may fail to achieve the system objective, or have undesirable side effects, e.g., the objective is achieved, but the autonomy of the agents is restricted unnecessarily. 

The increasing availability of large amounts of system behaviour data introduces the possibility of a new approach to the design of norms, namely the synthesis of norms directly from data collected during the execution of the system. 
For example, data may show that collisions always happen when an agent's speed is very high, allowing us to state a norm prohibiting agents from speeding too much.
%
In this paper, we consider the problem of synthesising norms from traces of agent behaviour, where each trace is labelled with whether the behaviour satisfies the system objective. We show that synthesising a norm (i.e., an obligation or a prohibition) that correctly classifies the traces (i.e., the norm is violated on traces where the behaviour does not satisfy the system objective, and is not violated on other traces) is an NP-complete problem. We also consider the problem of synthesizing a norm that is ``close'' to a target norm. 
This problem is relevant where there is an existing norm that does not achieve the system objective, but which is accepted, e.g., by human users of a system, and we wish a minimal modification that does achieve the objective. We show that the minimal norm revision problem is also NP-complete.

There has been a considerable amount of work on the automated synthesis of norms. For example, \cite{shoham1995social} consider the problem of synthesising a  \emph{social law} that constrains the behaviour of the agents in a MAS so as to ensure that agents in a \emph{focal} state are always able to reach another focal state no matter what the other agents in the system do. They show that synthesising a useful social law is NP-complete. 
In \cite{DBLP:journals/synthese/HoekRW07}, the problem of synthesising a social law is recast as an ATL model checking problem. They show that the problem of whether there exists a social law satisfying an objective expressed as an arbitrary ATL formula (feasibility) is NP-complete, while for objectives expressed as propositional formulae, feasibility (and synthesis) is decidable in polynomial time.
In contrast to the approach we present here, these and related approaches (e.g., \cite{FITOUSSI200061,WOOLDRIDGE2005396,DBLP:conf/ijcai/AgotnesHRSW07}) 
are used to synthesise norms 
at design time, and 
assume a complete model the agents' behaviour is available, e.g., in the form of a transition system or a Kripke structure.
Another strand of work focusses on the `online' synthesis of norms, where norms emerge from the interactions of agents in a decentralised way,  e.g., ~\cite{airiau2014emergence,savarimuthu2011norm}. Unlike our approach, such approaches typically assume that the agents are cooperative, and/or that some minimal standards of behaviour can be assumed. 
%
Closer to our work are approaches that use agents' behaviour to guide (centralised) norm synthesis. For example, \cite{morales2015online} present  algorithms for the online synthesis of compact \textit{action-based} norms when the behaviour of agents leads to undesired system states. 
In contrast, we consider norms that regulate patterns of behaviour.
%
\cite{DBLP:journals/tplp/CorapiRVPS11} and \cite{DBLP:conf/aamas/AthakraviCRVPS12} apply Inductive Logic Programming to norm synthesis and revision. They represent the system objective in terms of desired outcomes of event traces. In contrast to our approach, traces encode use cases provided by a human developer, and the system objectives are expressed in the same language as norms.

\section{Preliminaries}


In this section we give formal definitions of the behaviour of agents in the MAS 
and of conditional norms.

We assume a finite propositional language $L$ that contains propositions corresponding to properties of states of the MAS. A \emph{state} of the MAS 
is a propositional assignment. 
A conjunction of all literals (propositions or their negations) in a state $s$ will be referred to as a \emph{state description} of $s$. A propositional formula is a boolean combination of propositional variables. The definition of a propositional formula $\phi$ being true in a state $s$ is standard. 

A \emph{trace} is a finite sequence of states. 
We use the notation $(s_1,\ldots,s_k)$ for a trace consisting of states $s_1,\ldots,s_k$. We denote the $i$-\textit{th} state in a trace $\rho$ by $\rho[i]$. We assume that the behaviour exhibited by the agents in the MAS is represented by a set of finite traces $\Gamma$.
For example, a trace could be generated by the actions of all vehicles involved in a traffic accident.
We denote by $S(\Gamma)$ or simply by $S$ the set of states occurring in traces in $\Gamma$. Each subset $X$ of $S(\Gamma)$ is definable by
a propositional formula $\phi_X$ (a disjunction of state descriptions of states in $X$).  Note that the size of $\phi_X$ is linear in the size of $X$ (the sum of sizes of state descriptions of states in $X$). 
$\Gamma$ is partitioned into two sets $\Gamma_T$ (positive traces) and $\Gamma_F$ (negative traces). The partition is performed with respect to the system objective, which typically does not correspond directly to the properties expressible in $L$.  

The problem we wish to solve is how to generate a conditional norm which is
expressed using propositions from $L$ that is obeyed on traces in $\Gamma_{T}$ and violated on traces in $\Gamma_{F}$.

\begin{definition}[Conditional Norm]
A \emph{conditional norm} (over $L$) is a tuple $(\phi_C, Z(\phi_Z), \phi_D)$, where 
$\phi_C$, $\phi_Z$ and $\phi_D$ are propositional formulas over $L$, and
$Z\in\{P,O\}$ indicates whether the norm is a prohibition ($P$) or an obligation ($O$). 
\end{definition}
We refer to $\phi_C$ as the (detachment) condition of the norm, and $\phi_D$ as the deadline. $\phi_Z$ characterizes a state that is prohibited (resp.\ obligated) to occur after a state where the condition of the norm $\phi_C$ holds, and before a state where the deadline of the norm $\phi_D$ holds. We define the conditions for violation of norms formally below.

\begin{definition}[Violation of Prohibition]
A conditional prohibition $(\phi_C, P(\phi_P), \phi_D)$ is violated on a trace $(s_1,s_2,\ldots,s_m)$ if there are $i,j$ with $1 \leq i \leq j \leq m$ such that $\phi_C$ is true at $s_i$, $\phi_P$ is true at $s_j$, and there is no $k$ with $i < k < j$ such that $\phi_D$ is true at $s_k$.
\label{def:cond_norms}
\end{definition}

\begin{definition}[Violation of Obligation]
A conditional obligation $(\phi_C, O(\phi_O), \phi_D)$ is violated on a trace $(s_1,s_2,\ldots,s_m)$ if there are $i,j$ with $1 \leq i \leq j \leq m$ such
that $\phi_C$ is true at $s_i$, $\phi_D$ is true at $s_m$, and there is no
$k$ with $i \leq  k \leq j$ such that $\phi_O$ is true at $s_k$.
\end{definition}

A conditional norm is obeyed on a trace if it is not violated on that trace.
Violation conditions of conditional  norms can be expressed in Linear Time Temporal Logic (LTL) and evaluated on finite traces in linear time~\cite{DBLP:conf/atal/AlechinaDL14}.

\section{Complexity of Norm Synthesis}
\label{sec:complexity}

Given a set of agent behaviour traces $\Gamma$ partitioned into $\Gamma_T$ and $\Gamma_F$, we wish to synthesize a norm that correctly classifies each trace (that is, the norm is violated on all traces in $\Gamma_{F}$, and is not violated on any trace in $\Gamma_{T}$). Clearly, this is not always possible; two sets of traces may not be distinguishable by a single conditional norm (or even by a set of conditional
norms). For example: 
\[
\Gamma_{T} = \{ (s_1, s_2, s_3) \},\: \Gamma_{F} = \{ (s_1, s_1, s_2, s_3) \}
\]
cannot be distinguished by a conditional norm.

\subsection{Prohibition Synthesis}

Next we define formally the decision problem we call \emph{prohibition synthesis}.

\begin{definition}
The \emph{prohibition synthesis problem} is the following decision problem:
\begin{description}
\item[Instance] A finite set of propositions $L$;
a finite set of finite traces $\Gamma$ partitioned into $\Gamma_T$ and $\Gamma_F$, each trace given as a sequence of state
descriptions over $L$.
\item[Question] Are there three propositional formulas $\phi_C$, $\phi_P$, and $\phi_D$ over $L$ such that
\begin{description}
\item[Neg] every trace in $\Gamma_F$ violates $(\phi_C, P(\phi_P), \phi_D)$
\item[Pos] no trace in $\Gamma_T$ violates $(\phi_C, P(\phi_P), \phi_D)$
\end{description}
\end{description}
\end{definition}



The correspondence between sets of states and formulas over $L$ allows us to restate the prohibition synthesis problem as follows: given a set of positive traces $\Gamma_T$ and negative traces $\Gamma_F$, find three sets of states $X_C$, $X_P$, $X_D$
such that:
\begin{description}
\item[Neg] For every trace $\rho \in \Gamma_F$, exist $i$ and $j$ with $i \leq j$ such that
$\rho[i] \in X_C$, $\rho[j] \in X_P$, and there is no $k$ with $i < k < j$ such that $\rho[k] \in X_D$.
\item[Pos] For every trace $\rho \in \Gamma_T$, if for some $i$ and $j$, $i \leq j$, $\rho[i] \in X_C$,
$\rho[j] \in X_P$, then there exists $k$ such that $i < k < j$ and $\rho[k] \in X_D$.
\end{description}

\begin{theorem}\label{thm:prohibition-synthesis}
The prohibition synthesis problem is NP-complete.
\end{theorem}

\begin{proof}
The prohibition synthesis problem is clearly in NP (a non-deterministic Turing machine can guess the sets and check in polynomial time that they satisfy the conditions). To prove that it is NP-hard, we reduce 3SAT (satisfiability of a set of clauses with 3 literals) to prohibition synthesis.

Suppose an instance of 3SAT (a set of clauses $C_1,\ldots,C_n$ over variables $x_1,\ldots,x_m$) is given. We generate an instance of the prohibition synthesis problem such that it has a solution iff $C_1,\ldots,C_n$ are satisfiable (each clause contains at least one true literal).
The set of states in the prohibition synthesis problem consists of two states $s$ and $t$
which are a technical device (intuitively they serve as the detachment condition and the violation of the prohibition), and for each variable $x_i$, two states $u_i$ and $v_i$, intuitively meaning that $x_i$ is true ($u_i$) or false ($v_i$). Below is
the rest of the construction. Comments in square brackets explain the
intuition behind it.

The set of negative traces $\Gamma_F$ contains:
\begin{itemize}
\item a two state trace $(s,t)$ [together with $s,t \not \in X_C \cap X_P$ below, this forces $s \in X_C$ and
$t \in X_P$];
\item for every variable $x_i$ in the input, a trace
 $(s, v_i, t, s, u_i, t)$ [this ensures that either
$v_i$ or $u_i$ are not in $X_D$].
\end{itemize}

The set of positive traces $\Gamma_T$ contains:
\begin{itemize}
\item a single state trace $(s)$ [so $s$ cannot be in $X_C \cap X_P$];
\item $(t)$ [so $t$ cannot be in $X_C \cap X_P$];
\item for every variable $x_i$ in the input:
$(s, v_i, u_i, t)$ [this means that either $v_i$ or $u_i$ are in $X_D$]; $(v_i)$; $(u_i)$; $(v_i,t)$; $(u_i,t)$; $(s,v_i)$;  $(s,u_i)$;
\item for every pair of variables $x_i$, $x_j$ in the input:
 $(v_i, u_j)$;
$(u_j, v_i)$ [this together with preceding traces ensures that $v_i$ and $u_i$ are not in $X_C$ or $X_P$];
\item for each clause $C$ in the input over variables $x_j,x_k,x_l$:
$(s, z_j$, $z_k, z_l, t)$ where $z_i$ is $u_i$ if $x_i$ occurs in $C$ positively, and
$v_i$ if it occurs negatively.
\end{itemize}

The reduction is polynomial in the number of variables (quadratic) and clauses (linear).

We claim that there exists an assignment $f$ of $0,1$ to $x_1,\ldots,x_m$ satisfying
$C_1,\ldots,C_n$ if, and only if, there is a solution to the prohibition synthesis problem
above where $X_C = \{s\}$, $X_P=\{t\}$, and for every $i$, $u_i \in X_D$ iff $f(x_i)=1$
and $v_i \in X_D$ iff $f(x_i)=0$.

Assume an assignment $f$ satisfying $C_1,\ldots,C_n$ exists. Let $X_C=\{s\}$
and $X_P=\{t\}$. For every $i$, place $u_i$ in $X_D$ if $f(x_i)=1$ and $v_i \in X_d$ iff $f(x_i)=0$. This produces a solution because: $s,t$ satisfies \textbf{Neg}; for every $i$, either
$u_i$ or $v_i$ are not in $X_D$, so  $s, v_i, t, s, u_i, t$ satisfies \textbf{Neg}.
Positive traces satisfy \textbf{Pos}: either $s$ followed by $t$ does not occur on a trace, or
$u_i$, $v_i$ occur between $s$ and $t$ and one of them is in $X_D$, or (from the clause encoding) one
of the literals in the clause is true, so for positive $x_i$ it means that $u_i$ is in $X_D$
and \textbf{Pos} is satisfied, or for negative $\neg x_i$ it means that $v_i$ is in $X_D$ and
again \textbf{Pos} is satisfied.

Assume there is a solution to the prohibition synthesis problem. It is clear (see the comments in square brackets above) that it has to be of the form $X_C=\{s\}$, $X_P=\{t\}$ and $X_D$ containing some $u_i$s and $v_i$s.
In particular, since $(s, v_i, u_i, t)$ is a positive trace, for every $i$ either $u_i$ or $v_i$
has to be not in $X_D$. Set $f(x_i)$ to be 1 if $u_i$ in $X_D$ and 0 otherwise.
Then each clause $C= \{\sim x_j,\sim x_k,\sim x_l\}$ 
(where $\sim x_{j}$ denotes $x_j$ if it occurs positively or $\neg x_j$ if it occurs negatively) is satisfied by $f$ since for every
clause there will be one literal which is true. This is because $(s, z_j, z_k, z_l, t)$ is a
positive trace, and either for some positive literal $x_i$, $u_i$ is in $X_D$, or for some
negative literal $\neg x_i$, $v_i$ is in $X_D$, so $u_i$ is not in $X_D$, so $f(\neg x_i)=1$.
\end{proof}

\subsection{Obligation Synthesis}

We now consider the \emph{obligation synthesis problem}.

\begin{definition}
The \emph{obligation synthesis problem} is the following decision problem:
\begin{description}
\item[Instance] A finite set of propositions $L$,
a finite set $\Gamma$ of finite traces partitioned into $\Gamma_T$ and $\Gamma_F$, 
where each trace is given as a sequence of state
descriptions.
\item[Question] Are there three propositional formulas $\phi_C$, $\phi_O$, and $\phi_D$ over $L$ such that
\begin{description}
\item[Neg] every trace in $\Gamma_F$ violates $(\phi_C, O(\phi_O), \phi_D)$
\item[Pos] no trace in $\Gamma_T$ violates $(\phi_C, O(\phi_O), \phi_D)$
\end{description}
\end{description}
\end{definition}

Analogously to the prohibition synthesis problem, the obligation synthesis problem can be equivalently restated in terms of states: are there three sets of states $X_C$, $X_O$ and $X_D$ such that:
\begin{description}
\item[Neg] For every trace $\rho \in \Gamma_F$, there exist $i$ and $j$ with $i \leq j$ such that
$\rho[i] \in X_C$, $\rho[j] \in X_D$, and there is no $k$ with $i \leq  k \leq j$ such that $\rho[k] \in X_O$
\item[Pos] For every trace $\rho \in \Gamma_T$, if for some $i$ and $j$, $i \leq j$, $\rho[i] \in X_C$,
$\rho[j] \in X_D$, then there exists $k$ such that $i \leq k \leq j$ and $\rho[k] \in X_O$.
\end{description}

\begin{theorem}\label{thm:obligation-synthesis}
The obligation synthesis problem is NP-complete.
\end{theorem}
\begin{proof}
The obligation synthesis problem is clearly in NP.
To prove that it is NP-hard, we again use a reduction from the 3SAT problem.

As before, consider a set of clauses $C_1,\ldots,C_n$ over variables $x_1,\ldots,x_m$, which is an instance of  3SAT. We generate an instance of the obligation synthesis problem such that it has a solution iff $C_1,\ldots,C_n$ are satisfiable.
The idea of the reduction is similar to that for prohibitions. 
We use two auxiliary states $s$ and $t$, intuitively to serve as the detachment condition and the deadline, and make sure that neither  of them is also the obligation, 
but now instead of inserting a deadline between $s$ and $t$ in positive traces, we insert an obligation.
We want to make some subset of $\{v_i: i \in [1,...m]\} \cup \{u_i: i \in [1,...m]\}$ to be the obligation ($X_O$), so that exactly one of $v_i$, $u_i$ for each $i$ is in $X_O$.
Then $u_i \in X_O$ can encode that $x_i$ is true, and $v_i \in X_O$ that $x_i$ is false, and
we can make the encoding work by creating a positive trace corresponding to each
clause so that at least one of the literals in the clause should be true.

The set of negative traces contains: 
\begin{itemize}
\item a 2 state trace $(s,t)$ [this forces either \\
$s \in X_C \cap \overline{X_D} \cap \overline{X_O}$, $t \in X_D \cap \overline{X_C} \cap \overline{X_O}$, or \\
$s \in X_C \cap X_D \cap \overline{X_O}$, or \\
$t \in X_C \cap X_D \cap \overline{X_O}$. 

To rule out the latter two possibilities, we
require below that $s$ and $t$ on their own are positive traces.]  
\item for every variable $x_i$ in the input, a trace $(s, v_i, t, s, u_i, t)$ [this ensures that either $v_i$ or $u_i$ are not in $X_O$, because there is one $(s,..,t)$ sub-trace that does not contain a state from $X_O$].
\end{itemize}
The set of positive traces contains:
\begin{itemize}
\item a one state trace $(s)$ [so $s$ cannot be in $X_C \cap X_D \cap \overline{X_O}$]
\item a one state trace $(t)$ [so $t$ cannot be in $X_C \cap X_D \cap \overline{X_O}$]
\item for every variable $x_i$ in the input, a trace
$(s, v_i, u_i, t)$ [this ensures that either $v_i$ or $u_i$ are in $X_O$]
\item for each clause $C$ in the input over variables $x_j,x_k,x_l$, a 
trace
$(s, z_j, z_k, z_l, t)$ where $z_i$ is $u_i$ if $x_i$ occurs in $C$ positively, and
$v_i$ if it occurs negatively.  
\end{itemize}

The reduction is linear in the number of variables and clauses.

We claim that there exists an assignment $f$ of $0,1$ to $x_1,\ldots,x_m$ 
satisfying $C_1,\ldots,C_n$ if, and only if, there is a solution to the obligation synthesis problem 
above where $s \in X_C$, $t \in X_D$, and for every $i$, $u_i \in X_O$ iff $f(x_i)=1$ 
and $v_i \in X_O$ iff $f(x_i)=0$.
The proof of this claim is analogous to that of Theorem \ref{thm:prohibition-synthesis}.

Assume an assignment $f$ satisfying $C_1,\ldots,C_n$ exists. Let $X_C=\{s\}$
and $X_D=\{t\}$. For every $i$, place $u_i$ in $X_O$ iff $f(x_i)=1$ and $v_i \in X_O$ iff $f(x_i)=0$. It is easy to check that this is a solution to the obligation synthesis problem.


Assume there is a solution to the obligation synthesis problem. It is clear (see the comments in brackets above) that any solution should satisfy
$s \in X_C \cap \overline{X_D} \cap \overline{X_O}$ and $t \in X_D \cap \overline{X_C} \cap \overline{X_O}$.
Since $(s, v_i, t, s, u_i, t)$ is a negative trace for every $i$, this means that
it contains an unsatisfied conditional obligation.
This means that for every
$i$, either $v_i$ or $u_i$ is not in $X_O$.
%
Since $(s, v_i, u_i, t)$ is a positive trace, then in any solution, for every $i$, 
either $u_i$ or $v_i$ has to be in $X_O$. 
Hence we can use the membership in $X_O$ to produce a boolean valuation of
variables $x_i$ (1 if $u_i \in X_O$, and 0 if $v_i \in X_O$).
Since for every clause $C= \{\sim x_j, \sim x_k, \sim x_l\}$, the trace
$(s, z_j, z_k, z_l, t)$ (where $z_i$ is $v_i$ if $\sim x_i = \neg x_i$, and
$u_i$ if $\sim x_i = x_i$) is a positive trace, at least one of $z_i$ is
in $X_O$. This means that the valuation based on the membership in
$X_O$ 
satisfies all the clauses (since at least one literal in each clause will evaluate to 1).
\end{proof}

From Theorems \ref{thm:prohibition-synthesis} and \ref{thm:obligation-synthesis}, it follows immediately that the problem of generating a set of norms of size at most $m$ that correctly classifies a set of traces $\Gamma$ (that is, no norm $N_{i}$ is violated on any trace in $\Gamma_{T}$, and all traces in $\Gamma_{F}$ violates at least one norm $N_{j}$) is also NP-complete. 

\section{Complexity of Minimal Revision}

In this section, we consider the problem of (minimally) revising conditional prohibitions (we omit the treatment of obligations, which is analogous).
This problem is often relevant when there is an existing norm that does not achieve the system objective, and we wish a minimal modification of the existing norm that does achieve the objective. 

Assume we are given a set of traces and a conditional prohibition $(\phi_C, P(\phi_P), \phi_D)$, and need to change it in a minimal way so that it classifies the traces correctly. 
The editing distance between conditional prohibitions can be defined in various ways, e.g., for formulas $\phi_C, \phi_P, \phi_D$ in disjunctive normal form, this could be the sum of the numbers of added and removed disjuncts for all three formulas. 
Note that the set of non-equivalent propositional formulas built from the set $L$ is finite, and so is the number of possible different conditional prohibitions (or obligations). Regardless of how the distance between different conditional norms is defined, for a fixed set of propositional variables $L$ there is a maximal editing distance $max(L)$ between any two norms using formulas over $L$. 

Given some distance measure $dist$ defined for any two conditional prohibitions $\alpha_1$ and $\alpha_2$ over $L$, the decision problem for minimal revision 
can be stated as:

\begin{definition}
The (decision form) of the minimal prohibition revision problem is as follows:
\begin{description}
\item[Instance] A number $m$; a finite set of propositions $L$;
a finite set $\Gamma$ of finite traces partitioned into $\Gamma_T$ and $\Gamma_F$; a conditional prohibition $(\phi_C, P(\phi_P), \phi_D)$ over $L$.
\item[Question] Are there three propositional formulas $\phi'_C$, $\phi'_P$, and $\phi_D'$ over $L$ such that
\begin{description}
\item[Dist] $dist((\phi_C, P(\phi_P), \phi_D), (\phi'_C, P(\phi'_P), \phi'_D)) \leq m$
\item[Neg] every trace in $\Gamma_F$ violates $(\phi'_C, P(\phi'_P), \phi'_D)$
\item[Pos] no trace in $\Gamma_T$ violates $(\phi'_C, P(\phi'_P), \phi'_D)$
\end{description}
\end{description}
\end{definition}

\begin{theorem}
Let $dist(\alpha_1,\alpha_2)$ be computable in time polynomial in the size of
$\alpha_1$ and $\alpha_2$, and the range of $dist$ over norms built over propositions from $L$ be bounded by $max(L)$.
Then the minimal prohibition revision problem is NP-complete.
\end{theorem}
\begin{proof}
The membership in NP follows from the fact that a solution can be guessed and checked in polynomial time.

NP-hardness is by reduction from the prohibition synthesis problem. Note that if a solution to the prohibition synthesis problem exists, it will be at most at distance $max(L)$ from the input norm. So to solve the prohibition synthesis problem, we can ask for a solution to the minimal prohibition revision problem with $m = max(L)$.
\end{proof}

\section{Conclusions}
\label{sec:concl}

We considered the problem of synthesising and minimally revising norms to achieve a system objective from labelled traces of agent behaviour.  We showed that the problems of norm synthesis and revision are  NP-complete. In future work, we plan to investigate the synthesis of approximate norms (i.e., norms that do not classify all traces perfectly), and more tractable heuristic approaches to norm synthesis and revision.

\bibliographystyle{kr}
\bibliography{bibliography}

\end{document}